%% file: main.tex
\documentclass[letterpaper, 10 pt, conference]{ieeeconf}
\IEEEoverridecommandlockouts
\pdfoutput=1

\usepackage{times} 
\usepackage{enumerate}
\usepackage{mathtools}
\usepackage{amsmath, amssymb}
 
\usepackage{amsthm}
\usepackage{graphicx}
\usepackage{xcolor}
\usepackage{cite}

\usepackage{abbreviations}
\allowdisplaybreaks

\usepackage{siunitx}
\usepackage{booktabs, makecell, multirow}
\usepackage{aligned-overset}
\usepackage[hidelinks]{hyperref}
\usepackage{nicefrac}
\usepackage[labelformat=simple]{subfig}

\usepackage{tikz,pgfplots}
\pgfplotsset{compat=newest}

\newtheorem{theorem}{Theorem}
\newtheorem{lemma}[theorem]{Lemma}

\newtheorem{remark}[theorem]{Remark}
\newtheorem{assumption}[theorem]{Assumption}

\definecolor{colorSafEDMD}{HTML}{581845}
\definecolor{colorSafEDMDnI}{HTML}{C70039}
\definecolor{colorkEDMD}{HTML}{FFC300}

\newcommand\copyrighttext{%
	\footnotesize 
	\textcopyright 2025 IEEE. Personal use of this material is permitted.
	Permission from IEEE must be obtained for all other uses, in any current or future media, including reprinting/republishing this material for advertising or promotional purposes, creating new collective works, for resale or redistribution to servers or lists, or reuse of any copyrighted component of this work in other works.
}
\newcommand\copyrightnotice{%
	\begin{tikzpicture}[remember picture,overlay]
		\node[anchor=south,yshift=10pt] at (current page.south) {\fbox{\parbox{\dimexpr\textwidth-\fboxsep-\fboxrule\relax}{\copyrighttext}}};
	\end{tikzpicture}%
}

\title{\LARGE \bf
Kernel-based error bounds of bilinear Koopman surrogate models for nonlinear data-driven control 
}

\author{Robin Strässer$^{1}$, Manuel Schaller$^{2}$, Julian Berberich$^{1}$, Karl Worthmann$^{3}$, Frank Allgöwer$^{1}$
\thanks{F.\ Allgöwer is thankful that this work was funded by the Deutsche Forschungsgemeinschaft (DFG, German Research Foundation) under Germany's Excellence Strategy -- EXC 2075 -- 390740016 and within grant AL 316/15-1 -- 468094890.
K.\ Worthmann gratefully acknowledges funding by the German Research Foundation (DFG, project ID 545246093). 
R.\ Strässer thanks the Graduate Academy of the SC SimTech for its support.}
\thanks{$^{1}$R.\ Strässer, J.\ Berberich, and F.\ Allgöwer are with the Institute for Systems Theory and Automatic Control, University of Stuttgart, 70550 Stuttgart, Germany
	(e-mail: {\tt [straesser, berberich, allgower]@ist.uni-stuttgart.de}).}%
\thanks{$^{2}$M.\ Schaller is with the Faculty of Mathematics, Chemnitz University of Technology, 09111 Chemnitz, Germany
	(e-mail: {\tt manuel.schaller@math.tu-chemnitz.de}).}%
\thanks{$^{3}$K.\ Worthmann is with the Optimization-based Control Group, Institute of Mathematics, Technische Universität Ilmenau, 99693 Ilmenau, Germany
	(e-mail: {\tt karl.worthmann@tu-ilmenau.de}).}%
}

\begin{document} 
 
\maketitle
\copyrightnotice
\thispagestyle{empty}
\pagestyle{empty}


\begin{abstract}
    We derive novel deterministic bounds on the approximation error of data-based bilinear surrogate models for unknown nonlinear systems. 
    The surrogate models are constructed using kernel-based extended dynamic mode decomposition to approximate the Koopman operator in a reproducing kernel Hilbert space. 
    Unlike previous methods that require restrictive assumptions on the invariance of the dictionary, our approach leverages kernel-based dictionaries that allow us to control the projection error via pointwise error bounds, overcoming a significant limitation of existing theoretical guarantees. 
    The derived state- and input-dependent error bounds allow for direct integration into Koopman-based robust controller designs with closed-loop guarantees for the unknown nonlinear system. 
    Numerical examples illustrate the effectiveness of the proposed framework.
\end{abstract}

\input{sec1-introduction}
\input{sec2-preliminaries}
\input{sec3-bilinear-surrogate}
\input{sec4-numerics}
\input{sec5-conclusion}

\input{sec6-app-proof-error-bounds}
%
\bibliographystyle{IEEEtran}
\bibliography{literature}

\end{document}

%% file: sec1-introduction.tex
%
%
\section{INTRODUCTION}
Data-driven control has emerged as a prominent tool, especially for complex systems.
While there exist several successful applications of data-driven controllers for \emph{linear} systems with rigorous theoretical guarantees~\cite{waarde2023informativity,berberich2024overview}, extending these guarantees to nonlinear systems poses significant challenges~\cite{martin:schon:allgower:2023b}.
A promising framework for data-driven modeling and control of nonlinear systems is based on the Koopman operator~\cite{koopman:1931}, which is a linear infinite-dimensional operator in a lifted function space~\cite{bevanda:sosnowski:hirche:2021,brunton:budisic:kaiser:kutz:2022}.
When considering nonlinear systems with inputs, the Koopman operator gives rise to a bilinear system~\cite{iacob:toth:schoukens:2024}.
Practical implementations rely on finite-dimensional approximations of the Koopman operator that are typically constructed from data using regression techniques such as extended dynamic mode decomposition (EDMD)~\cite{williams:kevrekidis:rowley:2015}. 
Despite its 
success in various practical applications~\cite{budisic:mohr:mezic:2012,kim:quan:chung:2023}, obtaining rigorous theoretical guarantees for Koopman-based methods remains a central challenge. 
Therein, rigorous bounds on the approximation error play a crucial role in transferring guarantees from the surrogate model to the nonlinear system.
Recent works have established probabilistic error bounds for EDMD-based surrogate models~\cite{mezic:2022,zhang:zuazua:2023,nuske:peitz:philipp:schaller:worthmann:2023}, allowing the formulation of data-driven state-feedback controllers based on bilinear surrogates that come with closed-loop guarantees for the unknown nonlinear system~\cite{strasser:berberich:schaller:worthmann:allgower:2025}.
However, available error bounds, and hence closed-loop guarantees, often rely on restrictive assumptions on the projection error or are formulated using the $L_2$-norm~\cite{schaller:worthmann:philipp:peitz:nuske:2023} such that a \emph{pointwise} bound on the full approximation error of Koopman-based bilinear surrogates for controlled systems is missing.
To address the full approximation error, \cite{williams:rowley:kevrekidis:2016,klus:nuske:hamzi:2020} propose kernel EDMD (kEDMD) to approximate the Koopman operator corresponding to uncontrolled dynamics in an appropriately chosen reproducing kernel Hilbert space (RKHS).
More precisely, the error analysis proposed in~\cite{kohne:philipp:schaller:schiela:worthmann:2025} on the full approximation error of kEDMD is exploited in~\cite{bold:philipp:schaller:worthmann:2024} for approximants of controlled systems in the original state space.
However, since the dynamics are not lifted to a higher-dimensional space, the resulting surrogate is, in general, a nonlinear system, for which a robust controller design is challenging.

In this paper, we derive deterministic bounds on the full approximation error for a novel data-based bilinear surrogate model of an unknown nonlinear system.
The model is constructed using a kEDMD approximation of the Koopman operator based on a nonlinear kernel-based lifting. 
While existing approaches require unrealistic and hard-to-verify assumptions on (approximate) invariance of the dictionary, the proposed method employs a kernel-based dictionary only relying on invariance of the underlying RKHS.
Our bounds are deterministic, not relying on probabilistic sampling estimates, and pointwise, and thus offer qualitative insights into the dependence of the error on system properties. 
Further, we illustrate the merits of the proposed error bounds for data-driven control of nonlinear systems by designing Koopman-based robust controllers with end-to-end guarantees for a nonlinear example system, whose performance is illustrated with numerical simulations.

%% file: sec2-preliminaries.tex
%
%
\section{PRELIMINARIES}\label{sec:preliminaries}

We first introduce the problem setting (Section~\ref{sec:problem-setting}), before providing the necessary background on the RKHS (Section~\ref{sec:kernel-interpolation}) and the Koopman framework (Section~\ref{sec:Koopman-background}).

\subsection{Problem setting}\label{sec:problem-setting}
We consider the continuous-time nonlinear system 
\begin{equation}\label{eq:dynamics-nonlinear}
    \dot{x} 
    = f_c(x) + \sum\nolimits_{i=1}^{m}g_{ci}(x)u_i
    = f_c(x) + G_c(x) u
\end{equation}
with state $x\in\bbX\subset\bbR^n$, control input $u\in\bbU\subset\bbR^m$, drift dynamics $f_c:\bbR^n\to\bbR^n$, control dynamics $G_c: \bbR^n \to \bbR^{n \times m}$ with $G_c(x) = \begin{bmatrix}
    g_{c1}(x) & \cdots & g_{cm}(x)
\end{bmatrix}$, where $f_c$ and $G_c$ are unknown, but sufficiently smooth. 
$\bbX$ and $\bbU$ denote compact convex sets containing the origin in their interiors.
We assume $f_c(0)=0$ and local Lipschitz continuity of $f_c$, $G_c$ with constants $L_f$, $L_G$, respectively. 
Due to continuity, there exists $\bar{G}>0$ satisfying $\|G_c(x)\| \leq \bar{G}$ on the compact set $\bbX$.
Sampling system~\eqref{eq:dynamics-nonlinear} equidistantly in time for control inputs constant on each sampling interval, i.e., $u(t)\equiv u_k \in \bbU$ for all $t \in [t_k,t_k+\Delta t)$ with sampling period $\Delta t>0$ and $t_k=k \Delta t$, $k\in\bbN_0$, yields the discretized system dynamics
\begin{equation}\label{eq:dynamics-nonlinear-sampled}
    x_{k+1} 
    = x_k + \int_{t_k}^{t_k+\Delta t} f_c(x(t)) + G_c(x(t))u_k\,\dd t
\end{equation}
with $x(t) = x(t;x_0,u)$ and $x_k=x(t_k)$.
In the following, we abbreviate $(x_{k+1},x_k,u_k)$ by $(x^+,x,u)$ for ease of notation. 

The goal of this paper is to derive rigorous error bounds tailored to controller design for a novel data-driven bilinear surrogate of the (sampled) nonlinear system~\eqref{eq:dynamics-nonlinear-sampled}.
Here, we aim for a discrete-time surrogate since continuous-time surrogate models typically require state-derivative data, which is expensive and hard to obtain in practice.
Note that we consider the sampled system of a continuous-time system, as continuous-time representations often arise naturally from physical modeling.
Further, the relation to the underlying continuous-time system is necessary to bound the bilinearization error, which scales with $\Delta t$ (see Theorem~\ref{thm:error-bound}).
We choose a set of $d\in\bbN$ pairwise distinct data points $\cX = \{x_j\}_{j=1}^d \subset \bbX$ with $x_1 = 0$, where we define
$
    h_\cX = \sup_{x\in\bbX} \min_{x_j\in\cX} \|x - x_j \|
$
as the corresponding fill distance.
Note that $h_\cX \leq h_0$ for any $h_0>0$ if the data points in $\cX$ are evenly spaced within $\bbX$ with a distance of $\nicefrac{2h_0}{\sqrt{n}}$.
For each $x_j\in\cX$, we collect $d_j\geq m+1$ data triplets 
$
    \cX_{j} = \{x_j,u_{jl},x_{jl}^+\}_{l=1}^{d_j}
$
with control inputs $u_{jl}\in \bbU$, where $x_{jl}^+ = x(\Delta t;x_j,u_{jl})$, $l \in [1:d_j]$.\footnote{We denote an ordered set of integers $[a,b]\cap \bbZ$ by $[a:b]$.}
\begin{assumption}\label{ass:input-rank-condition}
    The data $\cX$ satisfies $h_\cX < 0.5$ and the control inputs~$\{u_{jl}\}_{l=1}^{d_j}$ satisfy for each $j\in[1:d]$, $d\geq 1$,
    \begin{equation}\label{eq:rank-condition-inputs}
        \rank\left(\bar{U}_j\right) = m + 1,
        \quad
        \bar{U}_j = \begin{bmatrix}
            1 & 1 & \cdots & 1 \\
            u_{j1} & u_{j2} & \cdots & u_{jd_j}
        \end{bmatrix}.
    \end{equation}
\end{assumption}
Assumption~\ref{ass:input-rank-condition} ensures that the applied input sequence is exciting enough to construct the data-driven surrogate model.
\begin{remark}
    Collecting $d_j$ data triplets for each $x_j\in\cX$ may be restrictive in practice. 
    However, we conjecture that the data collection for our proposed approach can be significantly weakened in future work.
    In particular, following the ideas in~\cite{schimperna:worthmann:schaller:bold:magni:2025}, it is possible to first collect data, e.g., in a trajectory-based manner, and then cluster the data such that each data point is in a small radius around an artificially introduced cluster center. 
    Further, Assumption~\ref{ass:input-rank-condition} may be relaxed by using the kernel-based approach in~\cite{bevanda:driessen:iacob:toth:sosnowski:hirche:2024}.
\end{remark}

\subsection{Reproducing kernel Hilbert space}\label{sec:kernel-interpolation}
Next, we introduce the necessary RKHS notation. 
We denote a symmetric and strictly positive definite kernel function by $\mathsf{k}:\bbR^n\times\bbR^n \to \bbR$ and define the \emph{canonical features} $\phi_x$ for $x\in\bbR^n$ by $\phi_x = \mathsf{k}(x,\cdot)$.
Further, for the data $\cX$ and the kernel $\mathsf{k}$ we define the kernel matrix
\begin{equation}\label{eq:kernel-matrix}
    K_\cX = \left[\begin{smallmatrix}
        \mathsf{k}(x_1,x_1) & \cdots & \mathsf{k}(x_d,x_1) \\
        \vdots & \ddots & \vdots \\
        \mathsf{k}(x_1,x_d) & \cdots & \mathsf{k}(x_d,x_d)   
    \end{smallmatrix}\right].
\end{equation}
Importantly, the kernel $\mathsf{k}$ induces a Hilbert space of functions $(\mathbb{H},\langle \cdot,\cdot\rangle_{\mathbb{H}})$ with reproducing property 
$
    f(x) = \langle f,\phi_x\rangle_{\mathbb{H}}
$
for all $x\in\bbR^n$ and $f\in\bbH$~\cite{wendland:2004}. 
Throughout the paper, we consider the RKHS $\bbH=\cN_s$ induced by piecewise-polynomial and compactly supported kernel functions based on the Wendland radial basis functions $\Theta_{n,s}:\bbR^n\to\bbR$ with state dimension $n\in\bbN$ and smoothness degree $s\in\bbN$.
More precisely, $\Theta_{n,s}(x) = \theta_{n,s}(\|x\|)$, where we refer to~\cite[Corollary 9.14]{wendland:2004} for the definition of $\theta_{n,s}$ for general $n$ and $s\in\{0,1,2,3\}$.
Importantly, $\theta_{n,s}\in \cC^{2s}([0,\infty),\bbR)$, i.e., it possesses $2s$ continuous derivatives~\cite[Lemma~9.8]{wendland:2004}.
As an example, for $n\leq 3$ and $s=1$, $\theta_{n,1}\in\cC^2([0,\infty),\bbR)$ is twice continuously differentiable and defined by
\begin{equation}
    \theta_{n,1}(r)=\begin{cases}
        (1-r)^4(4r+1) & \text{for } r<1, \\
        0 & \text{for } r\geq 1.
    \end{cases}
\end{equation}
Then, the induced Wendland kernel is given by 
$
    \mathsf{k}(x,y) = \theta_{n,s}(\|x-y\|)
$ 
for $x,y\in\bbR^n$.
\begin{assumption}\label{ass:smoothness-Wendland-kernel}
    The smoothness degree $s\in\bbN$ of the Wendland kernels satisfies $s\geq 1$.
\end{assumption}
This assumption ensures a sufficiently smooth kernel function, i.e., it is at least twice continuously differentiable.
\begin{lemma}\label{lm:bounded-hessian}
    Let Assumption~\ref{ass:smoothness-Wendland-kernel} hold. 
    Then, the canonical features $\phi_{x_j}$, $j\in[1:d]$, admit a local norm bound on their Hessian matrix on $\bbX$, i.e.,
    \newlength\mylenI
    \settoheight\mylenI{$
        \left\| \nabla^2 \phi_{x_j} (x)\right\| 
        = \left\| 
            \left[\begin{matrix}
                \nicefrac{\partial^2 \phi_{x_j}(x)}{\partial x_1^2}
                & \!\!\!\cdots\!\!\! & 
                \nicefrac{\partial^2 \phi_{x_j}(x)}{\partial x_1\partial x_n} \\
                \vdots & \!\!\!\ddots\!\!\! & \vdots \\                    
                \nicefrac{\partial^2 \phi_{x_j}(x)}{\partial x_n \partial x_1}
                & \!\!\!\cdots\!\!\! & 
                \nicefrac{\partial^2 \phi_{x_j}(x)}{\partial x_n^2}
            \end{matrix}\right]
        \right\| 
        \leq D_\phi
        \hspace*{0.1\linewidth}
    $}
    \begin{equation}\label{eq:assumption-second-derivative}
    \resizebox{\linewidth}{\mylenI}{$
        \hspace*{-0.03\linewidth}
        \left\| \nabla^2 \phi_{x_j} (x)\right\| 
        = \left\| 
            \left[\begin{matrix}
                \nicefrac{\partial^2 \phi_{x_j}(x)}{\partial x_1^2}
                & \!\!\!\cdots\!\!\! & 
                \nicefrac{\partial^2 \phi_{x_j}(x)}{\partial x_1\partial x_n} \\
                \vdots & \!\!\!\ddots\!\!\! & \vdots \\                    
                \nicefrac{\partial^2 \phi_{x_j}(x)}{\partial x_n \partial x_1}
                & \!\!\!\cdots\!\!\! & 
                \nicefrac{\partial^2 \phi_{x_j}(x)}{\partial x_n^2}
            \end{matrix}\right]
        \right\| 
        \leq D_\phi
        \hspace*{0.1\linewidth}
    $}
    \hspace*{-0.1\linewidth}
    \end{equation}
    for all $j\in[1:d]$ and $x\in\bbX$ with some $D_\phi>0$.
\end{lemma}
\begin{proof}
    Observe $\phi_{x_j}(x)=\mathsf{k}(x_j,x) = \theta_{n,k}(\|x_j-x\|)\in\cC^{2s}([0,\infty),\bbR)$ and $s\geq 1$ due to Assumption~\ref{ass:smoothness-Wendland-kernel}, i.e., the canonical features are at least twice continuously differentiable.
    Since the region $\bbX$ is compact, we directly conclude the norm bound~\eqref{eq:assumption-second-derivative} for some $D_\phi>0$.
\end{proof}
The native RKHS norm $\|\cdot\|_{\cN_s}$ is induced by the abstract inner product and thus typically hard to compute. For Wendland kernels, this norm is equivalent to a fractional Sobolev norm using an extension operator~\cite[Corollary~10.48]{wendland:2004}.

\subsection{Koopman operator}\label{sec:Koopman-background}
The Koopman operator is a linear but infinite-dimensional operator that views a dynamical system through observables $\psi:\mathbb{X}\to\mathbb{R}$~\cite{koopman:1931}. 
We define the Koopman operator $\cK$ corresponding to the discrete-time system $x^+ = F(x)$ by 
$
    (\cK \psi)(x) = \psi(x^+)
$.
In practice, a finite-dimensional matrix representation of the Koopman operator can be obtained based on data using regression-based tools such as EDMD.
More precisely, we rely on kEDMD for an induced Wendland kernel $\mathsf{k}$.
As shown in~\cite{kohne:philipp:schaller:schiela:worthmann:2025}, the propagation of an observable function $\psi\in\bbH$ can be approximated by
\begin{equation}\label{eq:psi-dynamics-kEDMD}
    \textstyle
    \psi(x^+) = (\cK \psi)(x) \approx \sum\nolimits_{j=1}^d (\hat{K} \psi_\cX)_j \phi_{x_j}(x),
\end{equation}
where $
        \psi_\cX = \begin{bmatrix}
            \psi(x_1) & \!\!\cdots & \!\!\psi(x_d)
        \end{bmatrix}^\top\!\!
$, $
    \hat{K} = K_\cX^{-1} K_{F(\cX)} K_\cX^{-1}
$, 
\begin{equation}\label{eq:kernel-matrix-F}
    K_{F(\cX)} = \left[\begin{smallmatrix}
        \mathsf{k}(x_1,F(x_1)) & \cdots & \mathsf{k}(x_d,F(x_1)) \\
        \vdots & \ddots & \vdots \\
        \mathsf{k}(x_1,F(x_d)) & \cdots & \mathsf{k}(x_d,F(x_d)) 
    \end{smallmatrix}\right],
\end{equation}
with pointwise error bound, cf.~\cite[Theorem~5.2]{kohne:philipp:schaller:schiela:worthmann:2025}.

%% file: sec3-bilinear-surrogate.tex
%
%
\section{BILINEAR KEDMD SURROGATE MODEL WITH DETERMINISTIC ERROR BOUNDS}\label{sec:bilinear-surrogate-with-error-bounds}

We introduce the proposed \emph{bilinear} surrogate model for the nonlinear system~\eqref{eq:dynamics-nonlinear-sampled} in Section~\ref{sec:data-driven-surrogate}, before analyzing the approximation error in Section~\ref{sec:error-bounds}.

\subsection{Data-driven surrogate model estimation}\label{sec:data-driven-surrogate}

To define the proposed surrogate model, we first estimate the nonlinear dynamics at the sampled data points~$\cX$ and then apply kEDMD. 
In particular, we construct surrogates approximating the forward Euler approximations
\begin{subequations}\label{eq:functions-discretized}
    \begin{align}
    \hspace*{-0.02\linewidth}
        f(x) &= x + \Delta t f_c(x),
        \
        \tilde{g}_i(x) = x + \Delta t (f_c(x) + g_{ci}(x)), 
        \\
    \hspace*{-0.02\linewidth}
        G(x) &= \begin{bmatrix}
            \tilde{g}_1(x) - f(x) & \cdots & \tilde{g}_m(x) - f(x)
        \end{bmatrix},
    \end{align}
\end{subequations}
at each data point $x_j \in \mathcal{X}$ based on the collected data $\cX_j$, $j\in[1:d]$.
To this end, we use the linear regression problem
\begin{equation}\label{eq:linear-regression}
    \argmin\limits_{H_j} \left\|
    \begin{bmatrix}
        x_{j1}^+ & \cdots & x_{jd_j}^+
    \end{bmatrix}
    - H_j
    \begin{bmatrix}
        1 & \cdots & 1 \\
        u_{j1} & \cdots & u_{jd_j}
    \end{bmatrix}
    \right\|
\end{equation}
to obtain $
    \hat{H}_j = \begin{bmatrix} 
        \hat{f}(x_j) & \hat{G}(x_j) 
    \end{bmatrix} \in \bbR^{n\times (m+1)}
$.
Here, the linear regression is well-posed due to Assumption~\ref{ass:input-rank-condition}, i.e.,~\eqref{eq:linear-regression} admits a unique solution~$\hat{H}_j\in \bbR^{n\times (m+1)}$, where our later established error analysis accounts for a possible mismatch $\hat{H}_j\approx \begin{bmatrix} 
    f(x_j) & G(x_j) 
\end{bmatrix}$ due to the approximation~\eqref{eq:functions-discretized} of the sampled dynamics~\eqref{eq:dynamics-nonlinear-sampled}.
Since $x_1=0$ is known to be the equilibrium of the nonlinear system, we explicitly encode $\hat{f}(x_1)=f(x_1)=0$, i.e., we set $\hat{H}_1=\begin{bmatrix} 0 & \hat{G}(x_j)\end{bmatrix}$.

Based on the estimated function values $\hat{H}_j$ at the data points in $\cX$, we may now derive the kEDMD-based \emph{bilinear} surrogate model.
To this end, we define the lifted and shifted state $\Psi(x) = \Phi(x) - \Phi(0)$, where
\begin{equation}\label{eq:features-Phi}
    \Phi(x) 
    =
    \begin{bmatrix}
        \phi_{x_1}(x) &
        \cdots &
        \phi_{x_d}(x)
    \end{bmatrix}^\top
\end{equation}
is a vector consisting of all canonical features corresponding to $\cX$.
Note that $\Psi(x)=0$ if and only if $x=0$ on $\bbX$ due to the required fill distance in Assumption~\ref{ass:input-rank-condition} and the support radius of the Wendland basis functions $\Theta_{n,s}$. 
Mimicking \eqref{eq:psi-dynamics-kEDMD}, we approximate the propagated lifted state by the surrogate
\begin{equation}\label{eq:dynamics-bilinear-surrogate-approx}
    \textstyle
    \Psi(x^+) 
    \approx A \Psi(x) 
    + B_0 u
    + \sum\nolimits_{i=1}^{m} u_i B_i \Psi(x),
\end{equation}
where
$
    A = K_{f(\cX)}^\top K_\cX^{-1}
$, $
    B_i = (K_{\tilde{g}_i(\cX)} - K_{f(\cX)})^\top K_\cX^{-1}
$, $
    B_0 = \begin{bmatrix}
        B_1\Phi(0) & \cdots & B_m\Phi(0)
    \end{bmatrix}
$, $K_\cX$ as in~\eqref{eq:kernel-matrix}, and $K_{F(\cX)}$ as in~\eqref{eq:kernel-matrix-F} for $F\in\{f,\tilde{g}_i\}$.
Here, we use the estimated values $\hat{H}_j$ of the unknown functions $f$, $\tilde{g}_i$, $i\in[1:m]$ at the data points $x_j\in\cX$ resulting from~\eqref{eq:linear-regression} for each $j\in[1:d]$.

\subsection{Rigorous error bounds}\label{sec:error-bounds}
Now, we state our main theorem, which bounds the approximation error of the data-driven surrogate~\eqref{eq:dynamics-bilinear-surrogate-approx}.
\begin{theorem}\label{thm:error-bound}
    Suppose Assumptions~\ref{ass:input-rank-condition} and~\ref{ass:smoothness-Wendland-kernel} hold.
    Then, there exist constants $C_1, C_2, h_0>0$ such that for $h_\cX\leq h_0$ any trajectory of the nonlinear system~\eqref{eq:dynamics-nonlinear-sampled} satisfies
    \begin{equation}\label{eq:dynamics-bilinear-surrogate}
        \Psi(x^+) 
        = A \Psi(x) 
        + B_0 u
        + \sum\nolimits_{i=1}^{m} u_i B_i \Psi(x)
        + r(x,u),
    \end{equation}
    where the residual $r(x,u)$ is bounded by
    \begin{equation}\label{eq:deterministic-bound-residual}
        \|r(x,u)\| 
        \leq 
        c_x \|x\| 
        + c_u \|u\|
        + c_{xx} \|x\|^2
        + c_{xu} \|x\|\|u\|
        + c_{uu}\|u\|^2
    \end{equation}
    for all $(x,u)\in\bbX\times\bbU$, where
    \begin{subequations}\label{eq:proportional-bound-residual-constants}
        \begin{align}
            c_x &= \tilde{u} (C_1h_\cX^{s-\nicefrac{1}{2}}\|\Phi\|_{\cN_s} + \sqrt{d} \Delta t^2 C_2 C_3 \|K_\cX^{-1}\|),
            \\
            c_u &= \Delta t^2 \left( \sqrt{md} C_2 C_3 \|K_\cX^{-1}\| + \sqrt{d} \tfrac{D_\phi}{2} \bar{G}^2 \right),
            \\
            c_{xu} &= \sqrt{m} C_1h_\cX^{s-\nicefrac{1}{2}}\|\Phi\|_{\cN_s} + 2 \sqrt{d} \Delta t^2 D_\phi L_f \bar{G},
            \\
            c_{xx} &= \sqrt{d} \Delta t^2 \tfrac{D_\phi}{2} L_f^2 (1 + \tilde{u} + \max_{u\in\bbU} \|u\|_1 ),
            \\ 
            c_{uu} &= \sqrt{d} \Delta t^2 \tfrac{D_\phi}{2} \bar{G}^2,
            \\
            \nonumber
            C_3 &= \tfrac{1}{2} (L_f \bar{x} + \bar{G} \bar{u}) (L_f + L_G \bar{u}) \max_{j\in[1:d]} \left\{\nicefrac{\sqrt{d_j}}{\sigma_{\min}(\bar{U}_j)}\right\}
        \end{align}
        with
        $
            \bar{x} = \max_{x\in\bbX}\{\|x\|\}
        $, $
            \bar{u} = \max_{u\in\bbU}\{\|u\|\}
        $, and $
            \tilde{u} = \max_{u\in\bbU} \big\{\big\|1 - \sum_{i=1}^m u_i\big\|\big\}
        $.
    \end{subequations}
\end{theorem}
\begin{proof}
    See~\ref{sec:app-proof-of-theorem}.
\end{proof}

Theorem~\ref{thm:error-bound} establishes the \emph{deterministic} error bound~\eqref{eq:deterministic-bound-residual} for the kEDMD-based \emph{bilinear} surrogate model~\eqref{eq:dynamics-bilinear-surrogate}.
Here, the constants $C_1$, $C_2$, and $h_0$ only depend on the considered domain $\bbX$ and are derived in~\cite[Section 3.3]{wendland:2004}.
Note that the constants $c_{x}$, $c_u$, $c_{xu}$, $c_{xx}$, $c_{uu}$ approach zero for $\Delta t \to 0$ and $d\to \infty$, where $\Delta t$ needs to converge to zero at a sufficiently fast rate since $\| K_\mathcal{X}^{-1} \|$ grows with $d$.
To establish the error bound~\eqref{eq:deterministic-bound-residual}, we exploit the kernel-based dictionary $\Psi$ based on the invariant native space of Wendland functions $\cN_s$, while existing error bounds in the literature require unrealistic and hard-to-verify assumptions on the (approximate) invariance of the dictionary and rely on probabilistic sampling estimates~\cite{strasser:schaller:worthmann:berberich:allgower:2025,strasser:schaller:worthmann:berberich:allgower:2024b}.
As a result, Theorem~\ref{thm:error-bound} bounds the full approximation error and offers qualitative insights into the dependence of the error on system properties and the collected data.
The bound is quadratic in state and input, and vanishes when approaching the origin.
Further,~\eqref{eq:deterministic-bound-residual} paves the way to rigorous closed-loop guarantees via robust controller design. 
\begin{remark}\label{rk:proportional_bounds}
    According to Theorem~\ref{thm:error-bound}, a robustly stabilizing controller for the uncertain surrogate dynamics~\eqref{eq:dynamics-bilinear-surrogate} is guaranteed to stabilize also the original nonlinear system.
    For the controller design,~\eqref{eq:deterministic-bound-residual} can be over-approximated to get an error bound which is \emph{proportional} in the state and input, i.e.,
    $
        \|r(x,u)\| \leq \tilde{c}_x \|x\| + \tilde{c}_u \|u\|
    $
    for the constants
    $
        \tilde{c}_x 
        = c_x + c_{xx} \bar{x} + c_{xu} \bar{u}
    $ and $
        \tilde{c}_u 
        = c_u + c_{uu} \bar{u}
    $.
    This bound can be directly employed to obtain rigorous closed-loop stability and performance guarantees via existing robust control~\cite{strasser:berberich:allgower:2025,strasser:berberich:schaller:worthmann:allgower:2025} and predictive control~\cite{worthmann:strasser:schaller:berberich:allgower:2024} techniques.
\end{remark}
The dimension of the \emph{bilinear} surrogate model~\eqref{eq:dynamics-bilinear-surrogate} scales with the number of samples in $\cX$ as we collect all canonical features $\phi_{x_j}$ in $\Psi$.
Developing a robust controller design based on a lower-dimensional bilinear surrogate model using, e.g., techniques from model order reduction, is left for future research.
If a bilinear surrogate is not required,~\cite{schimperna:worthmann:schaller:bold:magni:2025} proposes a nonlinear surrogate model for the lower-dimensional lifting $\Psi(x)=x$ to define a predictive controller.

%% file: sec4-numerics.tex
%
%
\section{NUMERICAL EXAMPLE}\label{sec:numerics}
We validate the derived bilinear surrogate model and show its combination with the recently proposed controller designs in~\cite{strasser:berberich:allgower:2025,strasser:berberich:schaller:worthmann:allgower:2025}.
To this end, we consider a zone temperature process used for building control~\cite{huang:2011}, which we have modified to be more nonlinear. 
In particular, we consider
$
    \dot{x} = V_z^{-1} u(T_0 \cos(\frac{1}{5}x) - x^3),
$
with zone temperature $x\in[x_{\min},x_{\max}]\subset\bbR$, air volume flow rate $u\in[u_{\min},u_{\max}]\subset\bbR$, zone volume $V_z$, supply air temperature $T_0$, where all variables denote deviations from the desired setpoints. 
The nonlinear dynamics are unknown, but uniformly gridded data samples $\cX$ in $\bbX$ are available, i.e., $\cX=\{x_{\min},...,-\delta,0,\delta,...,x_{\max}\}$ for $\delta = \nicefrac{(x_{\max}-x_{\min})}{(d-1)}$.
Further, we collect data $\cX_j$ for uniformly drawn inputs $\{u_{j1},u_{j2}\}\subset\bbX$, $j\in[1:d]$ with $V_z=2$, $T_0=-2$, $\bbX=[-1,1]$, $\bbU = [-2,2]$.
Note that the collected data satisfies Assumption~\ref{ass:input-rank-condition} for $d > 3$ data samples.

We construct a kEDMD-based bilinear surrogate model according to Section~\ref{sec:data-driven-surrogate} for the Wendland RKHS with smoothness degree $s=1$. 
To this end, we collect $d$ data points and use the corresponding canonical features to define the lifting function $\Psi$.
Figure~\ref{fig:prediction-error} shows the open-loop prediction error of the surrogate~\eqref{eq:dynamics-bilinear-surrogate-approx}, i.e., dynamics $\Psi_{t+1} = A\Psi_t + B_0 u + \sum_{i=1}^m u_i B_i \Psi_t$ with $\Psi_0=\Psi(x(0))$, and the true value $\Psi(x(t))$ over time, where we highlight the average error for different data lengths $d\in\{5,7,...,19\}$ and the achieved range from worst to best error.
For the open-loop simulation, the inputs are uniformly drawn from~$\bbU$ for initial condition $x(0)=0$.
We compare the surrogate~\eqref{eq:dynamics-bilinear-surrogate-approx} with a SafEDMD-based bilinear surrogate~\cite{strasser:schaller:worthmann:berberich:allgower:2024b}.
For the latter, we employ two different SafEDMD surrogates: 1) based on the kernel-based lifting function $\Psi$ and 2) based on a monomial dictionary of maximal degree three, i.e., $\Psi_\mathrm{mon}(x)=\begin{bmatrix}
    x & x^2 & x^3
\end{bmatrix}^\top$.
The main difference between the two approaches is that SafEDMD cannot rigorously characterize the resulting projection error, while our kEDMD-based surrogate comes with a full approximation error bound.
As shown in Figure~\ref{fig:prediction-error}, the prediction error of the kEDMD is at least as good as the one of SafEDMD using the same lifting function $\Psi$. 
Although SafEDMD with the monomial lifting $\Psi_\mathrm{mon}$ yields a smaller prediction error, it comes without a bound on the projection error and, thus, it has no guaranteed relation to the underlying nonlinear system.

The practical implications of this fact are now demonstrated by using the surrogate models for designing a data-driven controller for the nonlinear system via the sum-of-squares approach from~\cite[Corollary~4]{strasser:berberich:allgower:2025}. 
For the proportional error bound described in Remark~\ref{rk:proportional_bounds}, we consider $\tilde{c}_x=\tilde{c}_u=0.05$.
Figure~\ref{fig:closed-loop-error} shows the average performance for $d=5$ and $100$ uniformly drawn initial conditions within~$\bbX$, where each surrogate leads to a stabilizing controller for the nonlinear system.
Here, the kEDMD-based controller yields the fastest convergence to the origin and achieves the best performance w.r.t.\ the criterion $\sum_{k=0}^{\lfloor t/\Delta t\rfloor} \|x_k\| + \|u_k\|$, where a more detailed performance investigation is left for future research.
We note again that only the proposed surrogate with the pointwise error bounds in Theorem~\ref{thm:error-bound} leads to \emph{guaranteed} closed-loop stability of the unknown nonlinear system.
\bgroup
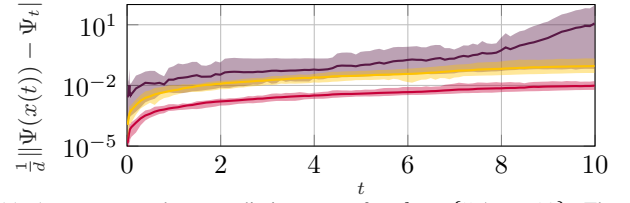
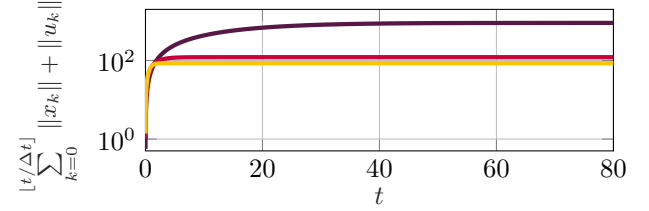
\begin{figure}[tb]
    \centering
    \captionsetup[subfloat]{captionskip=-1pt}
    \subfloat[Average open-loop prediction error for $d\in\{5,7,...,19\}$. The shaded area indicates the range from worst to best error.]{\label{fig:prediction-error}
    \shortstack{
        \input{prediction-error-max.tex}
    }}
    \\[-0.1\baselineskip]
    \subfloat[Average closed-loop performance for $100$ randomly drawn $x(0)\in\bbX$.]{\label{fig:closed-loop-error}
        \input{closed-loop-performance-max.tex}
    }
    \caption{Behavior of kEDMD surrogate~\eqref{plot:kEDMD} and SafEDMD surrogates based on $\Psi$~\eqref{plot:SafEDMD} and $\Psi_\mathrm{mon}$~\eqref{plot:SafEDMD-monomials}.}
    \vspace*{-\baselineskip}
\end{figure}
\egroup

%% file: prediction-error-max.tex
%
%
\begin{tikzpicture}
\begin{axis}[%
width=0.9\linewidth,
height=0.4\linewidth,
xmin=0,
xmax=10,
xlabel style={font=\color{white!15!black},yshift=3pt},
xlabel={\scriptsize$t$},
ymode=log,
ymin=1e-05,
ymax=100,
yminorticks=true,
ylabel style={font=\color{white!15!black}},
ylabel={\small\hspace*{-0.03\linewidth}$\frac{1}{d} \|\Psi(x(t)) - \Psi_t\|$},
xmajorgrids,
ymajorgrids,
yminorgrids,
]

\addplot[area legend, draw=none, fill=colorSafEDMDnI, fill opacity=0.4, forget plot]
table[row sep=crcr] {%
x	y\\
0	0\\
0.01	5.35294266595078e-07\\
0.02	3.26234682705093e-06\\
0.03	3.68340265260186e-06\\
0.04	1.41376961481556e-06\\
0.05	1.64624617180108e-05\\
0.06	1.47092861332234e-05\\
0.07	2.44398945828325e-05\\
0.08	2.3201394621179e-05\\
0.09	2.80870673862511e-05\\
0.19	4.65832518874422e-05\\
0.29	0.000107242628259579\\
0.39	0.000185212103408564\\
0.49	0.000251529531715632\\
0.59	0.00031340562062428\\
0.69	0.000386987241794155\\
0.79	0.000421718603953903\\
0.89	0.000546449906756754\\
0.99	0.000561251426067992\\
1.09	0.000600513302103331\\
1.19	0.000619599960685127\\
1.29	0.000649500325867813\\
1.39	0.000682999033172138\\
1.49	0.000674488987874679\\
1.59	0.000714321796952796\\
1.69	0.00078507880930739\\
1.79	0.000833816072547927\\
1.89	0.000902370438964063\\
1.99	0.00101165910203407\\
2.09	0.0010185705550681\\
2.19	0.00105826307784395\\
2.29	0.00114989886151948\\
2.39	0.00123545040087371\\
2.49	0.00126424434374854\\
2.59	0.00132457271805162\\
2.69	0.00130635166087353\\
2.79	0.00137232067404265\\
2.89	0.00141606394053949\\
2.99	0.00153335064577564\\
3.09	0.0016246017550384\\
3.19	0.00169910130277355\\
3.29	0.00170943656846224\\
3.39	0.0017026685434204\\
3.49	0.00172461470168436\\
3.59	0.00173119137893461\\
3.69	0.0018489664992071\\
3.79	0.00203060569207777\\
3.89	0.00209980776517727\\
3.99	0.00216320828104808\\
4.09	0.00219277520769306\\
4.19	0.00218841851009814\\
4.29	0.00228816927757834\\
4.39	0.00242530663542471\\
4.49	0.00247440777318961\\
4.59	0.0025657103044773\\
4.69	0.00264664612873616\\
4.79	0.00269295610889033\\
4.89	0.00281093431446011\\
4.99	0.00281431841770167\\
5.09	0.00285919325597658\\
5.19	0.00283034277168956\\
5.29	0.00288553543593336\\
5.39	0.00289911215916986\\
5.49	0.00300114853785725\\
5.59	0.00302535513270789\\
5.69	0.00310488991791849\\
5.79	0.00320398658706196\\
5.89	0.0032981379217935\\
5.99	0.00332882541937319\\
6.09	0.00340595042491848\\
6.19	0.00337548866477658\\
6.29	0.00342331842540828\\
6.39	0.00345826567059985\\
6.49	0.00353078816760516\\
6.59	0.00361154515911843\\
6.69	0.00371537050724636\\
6.79	0.00374103805479368\\
6.89	0.00381324723526714\\
6.99	0.00390527013839879\\
7.09	0.00402911432084059\\
7.19	0.00415471197303972\\
7.29	0.00414356761071803\\
7.39	0.00423431063388089\\
7.49	0.00422937434521291\\
7.59	0.0043464486295459\\
7.69	0.00437571560773071\\
7.79	0.0043569083499796\\
7.89	0.00449535251870814\\
7.99	0.00465477283698585\\
8.09	0.0047821221717575\\
8.19	0.0048212676005773\\
8.29	0.00488084168743379\\
8.39	0.00496805135389816\\
8.49	0.00508177062463575\\
8.59	0.00514287938921394\\
8.69	0.00516835199088048\\
8.79	0.00521775602201524\\
8.89	0.00533507540158897\\
8.99	0.0053479966361495\\
9.09	0.00546359888717878\\
9.19	0.00555848850014683\\
9.29	0.00565823791430382\\
9.39	0.00570671393090692\\
9.49	0.00578218078072481\\
9.59	0.00589058873109465\\
9.69	0.00590625121132886\\
9.79	0.00596662575349336\\
9.89	0.00600410535542408\\
9.99	0.00602760664889716\\
9.99	0.0163707033983884\\
9.89	0.0154291544036045\\
9.79	0.0155455099339195\\
9.69	0.0152808915355095\\
9.59	0.0163068522445927\\
9.49	0.0161080876562554\\
9.39	0.016616071908127\\
9.29	0.0155805734688507\\
9.19	0.0158775515647999\\
9.09	0.0161282505590855\\
8.99	0.0155130288628654\\
8.89	0.015352574841054\\
8.79	0.0150622009330558\\
8.69	0.0148786638292679\\
8.59	0.0141771730914063\\
8.49	0.0128344745708179\\
8.39	0.0141920071088444\\
8.29	0.0140828324473961\\
8.19	0.0147297246317556\\
8.09	0.0142955673091818\\
7.99	0.0138384379982244\\
7.89	0.0145942792978041\\
7.79	0.0140257296312442\\
7.69	0.0134633813235609\\
7.59	0.0138680572692916\\
7.49	0.014002933862885\\
7.39	0.0126803080469558\\
7.29	0.012189010179097\\
7.19	0.0128924406977124\\
7.09	0.0136454995348906\\
6.99	0.0122188837452193\\
6.89	0.0114855146803974\\
6.79	0.0107459047088226\\
6.69	0.00903357444626463\\
6.59	0.00888136943349124\\
6.49	0.00847692080649178\\
6.39	0.00901301142364618\\
6.29	0.00835603099695692\\
6.19	0.00848757803684037\\
6.09	0.00845318609200646\\
5.99	0.00841240767709181\\
5.89	0.00836121699083076\\
5.79	0.00758507395387806\\
5.69	0.00751765307667694\\
5.59	0.00752957741638483\\
5.49	0.00742136826162513\\
5.39	0.00687871379792554\\
5.29	0.00626010102950717\\
5.19	0.00616514216446313\\
5.09	0.00622695406028547\\
4.99	0.00588863825866965\\
4.89	0.00589324638062356\\
4.79	0.00552288457609893\\
4.69	0.00566292416012926\\
4.59	0.00566877485795633\\
4.49	0.00580832824585664\\
4.39	0.00567959981160017\\
4.29	0.00498804957521362\\
4.19	0.0048287633051375\\
4.09	0.00455990188559793\\
3.99	0.00434033416320599\\
3.89	0.00409390161109329\\
3.79	0.00440320333581252\\
3.69	0.00468166091951122\\
3.59	0.00436052443593295\\
3.49	0.0042150757322159\\
3.39	0.00427102931518309\\
3.29	0.0039019081650643\\
3.19	0.00402923756109811\\
3.09	0.00391059561898835\\
2.99	0.00345937175411372\\
2.89	0.00313563115477182\\
2.79	0.00315471000040744\\
2.69	0.00291122317623548\\
2.59	0.00296815260302547\\
2.49	0.00259805738775299\\
2.39	0.00244855511423755\\
2.29	0.00232748293501619\\
2.19	0.00232993722757204\\
2.09	0.0021788939799408\\
1.99	0.00222891666236297\\
1.89	0.00209789331666637\\
1.79	0.00191923840934475\\
1.69	0.00180277352605655\\
1.59	0.00171133566286136\\
1.49	0.00164673882854526\\
1.39	0.00139310383833877\\
1.29	0.00125048360773974\\
1.19	0.00130716182738208\\
1.09	0.001129811163846\\
0.99	0.00119114117016409\\
0.89	0.00117582677087325\\
0.79	0.000989162309734494\\
0.69	0.000930747147423704\\
0.59	0.000886520050557006\\
0.49	0.000653694377462956\\
0.39	0.000581197355917422\\
0.29	0.000382946633234975\\
0.19	0.000225625159396255\\
0.09	0.000150760061142795\\
0.08	0.000137238599424125\\
0.07	0.000123682759737153\\
0.06	0.000124278422819113\\
0.05	0.00013116648525736\\
0.04	9.46286436112698e-05\\
0.03	5.24095733610981e-05\\
0.02	3.79864002541919e-05\\
0.01	4.30750679972094e-05\\
0	0\\
}--cycle;
\addplot [color=colorSafEDMDnI, thick]
  table[row sep=crcr]{%
0.00999999999999979	1.37230239308455e-05\\
0.0199999999999996	2.0809914250826e-05\\
0.0299999999999994	2.6103409363554e-05\\
0.0399999999999991	4.04152334553748e-05\\
0.0600000000000005	6.65132652398537e-05\\
0.0899999999999999	8.54731609335531e-05\\
0.19	0.000144751588078397\\
0.289999999999999	0.000233584225934453\\
0.390000000000001	0.000325794100222136\\
0.49	0.000395915642971448\\
0.59	0.000504582788046581\\
0.69	0.000589301594544414\\
0.790000000000001	0.000643149401998533\\
0.890000000000001	0.000723841398598818\\
0.99	0.000787064641105501\\
1.09	0.000832174549061865\\
1.19	0.000937076968191899\\
1.29	0.0010283618207835\\
1.39	0.00107653570037968\\
1.49	0.00119571970565761\\
1.69	0.00141915695592992\\
1.79	0.00147469809166992\\
1.89	0.00156280607401985\\
2.39	0.00184817104829581\\
2.49	0.00194152099419945\\
2.59	0.00202321449642464\\
2.69	0.00208668841373079\\
2.79	0.00220727626908155\\
2.89	0.00222676533854101\\
2.99	0.00234809185482277\\
3.19	0.00256465458069557\\
3.29	0.00260329273607652\\
3.49	0.00280503497010786\\
3.79	0.00304385784416627\\
3.89	0.0030180489836391\\
3.99	0.00318797511928555\\
4.09	0.00321407182872142\\
4.29	0.00339183163983066\\
4.39	0.0035866391199029\\
4.59	0.00375212905066231\\
4.69	0.00377923827597611\\
4.79	0.00378013273574863\\
4.99	0.00391535367743596\\
5.19	0.0040695510626696\\
5.39	0.00426765160450263\\
5.79	0.00451195131980653\\
5.89	0.00462711998489494\\
6.09	0.00470679166058519\\
6.19	0.00482308942157823\\
6.29	0.00486535285830337\\
6.49	0.00518934938360599\\
6.59	0.00525918724407504\\
6.79	0.00567815464121209\\
6.89	0.0058304358981193\\
6.99	0.00606944594074412\\
7.09	0.00626977452580361\\
7.29	0.00633323536308278\\
7.39	0.00650414633011051\\
7.49	0.00681702937342485\\
7.89	0.00717619712128134\\
7.99	0.00718060556565399\\
8.09	0.00741926601125279\\
8.19	0.00753058928085915\\
8.29	0.00754798517055792\\
8.39	0.00781239162664706\\
8.49	0.00782281047793289\\
8.59	0.00831595998582228\\
8.69	0.00856809431019329\\
8.79	0.00846543078742343\\
8.89	0.0087143527980832\\
8.99	0.00884149853071794\\
9.39	0.00909008368435419\\
9.49	0.0092350503134453\\
9.59	0.00930533319156646\\
9.69	0.00920166138934436\\
9.79	0.00935699995507909\\
9.89	0.00938099953496469\\
9.99	0.00971318634034669\\
};

\addplot[area legend, draw=none, fill=colorSafEDMD, fill opacity=0.4, forget plot]
table[row sep=crcr] {%
x	y\\
0	0\\
0.01	5.58757944398186e-06\\
0.02	4.15447817382422e-05\\
0.03	0.000115174846500795\\
0.04	0.000227920927649268\\
0.05	0.000203869921486035\\
0.06	0.000314930638274407\\
0.07	0.000276869034069408\\
0.08	0.000310386800929008\\
0.09	0.000321794558917396\\
0.19	0.000501755262567266\\
0.29	0.000826759668568699\\
0.39	0.00146427733386179\\
0.49	0.00184726894395032\\
0.59	0.00219698923574499\\
0.69	0.00339871080376877\\
0.79	0.00452292029040186\\
0.89	0.00456305502722649\\
0.99	0.00420022132884191\\
1.09	0.00423877187337909\\
1.19	0.00415311893858807\\
1.29	0.00456068107183444\\
1.39	0.00507199370531484\\
1.49	0.0052760161812875\\
1.59	0.00552322271992494\\
1.69	0.00599679196500459\\
1.79	0.00636179229025216\\
1.89	0.00688183136474824\\
1.99	0.00779243616765067\\
2.09	0.0075096069677366\\
2.19	0.00755096455928614\\
2.29	0.00848902657837304\\
2.39	0.00907710933362971\\
2.49	0.0090458610428486\\
2.59	0.0101176041536125\\
2.69	0.0110324112485355\\
2.79	0.010813771657383\\
2.89	0.0120362265757138\\
2.99	0.0153916915558702\\
3.09	0.012139052238825\\
3.19	0.0128303273897124\\
3.29	0.0184970467462468\\
3.39	0.0150725771174548\\
3.49	0.0164599049174725\\
3.59	0.0156499770902324\\
3.69	0.0158229730294496\\
3.79	0.0158296840218333\\
3.89	0.0156934000509297\\
3.99	0.0220630607680095\\
4.09	0.0197019041567441\\
4.19	0.0201699250954778\\
4.29	0.0229337991823223\\
4.39	0.0234444838363729\\
4.49	0.0237904283318762\\
4.59	0.0259757289344916\\
4.69	0.0256718094040686\\
4.79	0.026241712952111\\
4.89	0.0263353326700295\\
4.99	0.0271666493241929\\
5.09	0.027827373294268\\
5.19	0.0283778354598724\\
5.29	0.0295497342506203\\
5.39	0.0293291763010988\\
5.49	0.029661802926146\\
5.59	0.0302428440974263\\
5.69	0.033029224646937\\
5.79	0.0330278255921204\\
5.89	0.0336768919721006\\
5.99	0.0345733218600261\\
6.09	0.0352328479525892\\
6.19	0.0350750524796772\\
6.29	0.0378853685252356\\
6.39	0.037619659646362\\
6.49	0.0358406489234014\\
6.59	0.0363435772884897\\
6.69	0.0369082307786511\\
6.79	0.0374912776017647\\
6.89	0.038665518179329\\
6.99	0.0392804986227594\\
7.09	0.0438203717249957\\
7.19	0.0450940066196156\\
7.29	0.047098018203932\\
7.39	0.0476230197990592\\
7.49	0.0452715001137511\\
7.59	0.0457409269108533\\
7.69	0.047550811857803\\
7.79	0.0490923462961221\\
7.89	0.0494623603012218\\
7.99	0.0515843311658332\\
8.09	0.0549041257495217\\
8.19	0.056731997556417\\
8.29	0.0577094499262997\\
8.39	0.0604366299036404\\
8.49	0.0593900116550533\\
8.59	0.0616667043033805\\
8.69	0.0630972852425136\\
8.79	0.0637866856311361\\
8.89	0.0624189802524943\\
8.99	0.0622326939754553\\
9.09	0.0637149414069461\\
9.19	0.062303450368765\\
9.29	0.0618767817124557\\
9.39	0.0634997762424093\\
9.49	0.0660332253694354\\
9.59	0.0605590700995853\\
9.69	0.0648583235211149\\
9.79	0.0662459134134526\\
9.89	0.065894838158392\\
9.99	0.0652694696515025\\
9.99	90.2504176463114\\
9.89	64.5817157475321\\
9.79	66.0263594774823\\
9.69	56.4767388620933\\
9.59	51.8767733281074\\
9.49	43.1556019361962\\
9.39	36.6773469699481\\
9.29	32.4063636416014\\
9.19	24.9352741536581\\
9.09	19.216104642992\\
8.99	14.2845887109632\\
8.89	12.4061450403311\\
8.79	10.0730004402345\\
8.69	8.60467339001824\\
8.59	6.92759644809373\\
8.49	5.28798344494889\\
8.39	4.74431042022924\\
8.29	3.63432406502149\\
8.19	2.66082968677659\\
8.09	3.80950641014721\\
7.99	2.27774278806515\\
7.89	1.20850731975369\\
7.79	1.13344627276108\\
7.69	1.20470897553112\\
7.59	1.56009845024363\\
7.49	1.02580437238196\\
7.39	1.17278456292985\\
7.29	0.934069933002727\\
7.19	0.954994294050493\\
7.09	0.732729522518015\\
6.99	0.597704785523287\\
6.89	0.700581071683149\\
6.79	0.699446729955699\\
6.69	0.993337894572451\\
6.59	1.00419519285847\\
6.49	0.922100382138597\\
6.39	0.766530440928854\\
6.29	0.772968523491294\\
6.19	0.871203613472426\\
6.09	0.843366518208519\\
5.99	0.813672740733937\\
5.89	0.794882883069552\\
5.79	0.802372421714322\\
5.69	0.770253739475495\\
5.59	0.647204993481256\\
5.49	0.554759481345044\\
5.39	0.647854991014195\\
5.29	0.589416721771366\\
5.19	0.502646740753162\\
5.09	0.489261104322813\\
4.99	0.45727778906264\\
4.89	0.338881164076816\\
4.79	0.414853773129531\\
4.69	0.424511174655311\\
4.59	0.417437766527405\\
4.49	0.379838728669494\\
4.39	0.240926806616144\\
4.29	0.216688865330946\\
4.19	0.215403386068951\\
4.09	0.214697014749335\\
3.99	0.216350808436494\\
3.89	0.219979380524497\\
3.79	0.220135052880466\\
3.69	0.213673036869559\\
3.59	0.217110000913653\\
3.49	0.213468964704769\\
3.39	0.213201082738898\\
3.29	0.214742548788087\\
3.19	0.224162037221687\\
3.09	0.219284442967626\\
2.99	0.219630832013171\\
2.89	0.214882067049722\\
2.79	0.2137131770681\\
2.69	0.213785724299032\\
2.59	0.217331422215437\\
2.49	0.217646921240491\\
2.39	0.22321493787413\\
2.29	0.231205377053214\\
2.19	0.222706871784023\\
2.09	0.210240964760789\\
1.99	0.136154401165894\\
1.89	0.156445376720034\\
1.79	0.110205208763526\\
1.69	0.0940550768985089\\
1.59	0.0980113266701856\\
1.49	0.146178770616291\\
1.39	0.145946600377776\\
1.29	0.122913874740749\\
1.19	0.111256329588708\\
1.09	0.121225846976317\\
0.99	0.0885939195517347\\
0.89	0.073873859327096\\
0.79	0.0283973545924387\\
0.69	0.0363524938861052\\
0.59	0.0522472744203667\\
0.49	0.052742981021545\\
0.39	0.0486766657536561\\
0.29	0.0377847719906176\\
0.19	0.0307021139120095\\
0.09	0.0130654710941239\\
0.08	0.0124962009884236\\
0.07	0.0117344002034624\\
0.06	0.0152073438568336\\
0.05	0.0653575946434166\\
0.04	0.0300501482583691\\
0.03	0.0151304973274988\\
0.02	0.0196927399928982\\
0.01	0.012422097498562\\
0	0\\
}--cycle;
\addplot [color=colorSafEDMD, thick]
  table[row sep=crcr]{%
0.00999999999999979	0.00258123450316623\\
0.0199999999999996	0.00306633949570032\\
0.0299999999999994	0.00331572474556392\\
0.0399999999999991	0.00651436828941966\\
0.0500000000000007	0.0101036148998748\\
0.0600000000000005	0.00343690455793002\\
0.0700000000000003	0.00317989712891802\\
0.0800000000000001	0.00382924780351305\\
0.0899999999999999	0.00362275234668615\\
0.19	0.00710075774033613\\
0.289999999999999	0.00958807437991911\\
0.390000000000001	0.0134805451173968\\
0.59	0.012284206516158\\
0.69	0.0134375107922645\\
0.790000000000001	0.0112162088123025\\
0.890000000000001	0.0192660023734583\\
0.99	0.0192877525282164\\
1.09	0.0230382377290682\\
1.19	0.0241033716287908\\
1.29	0.0257968241011506\\
1.39	0.0291789662982118\\
1.49	0.0319605038329018\\
1.59	0.0269361447939097\\
1.69	0.0271067075475531\\
1.79	0.0275028753456329\\
1.89	0.0359808906869874\\
1.99	0.0340022210358642\\
2.09	0.0440617686606741\\
2.19	0.0464234880396076\\
2.29	0.0475143365774185\\
2.39	0.0467448463594147\\
2.49	0.0449630235887021\\
2.59	0.0453760939514179\\
2.69	0.0451426700233816\\
2.79	0.0475089316794012\\
2.89	0.0466697628429791\\
3.09	0.0532825518058595\\
3.19	0.053648007573145\\
3.29	0.0513998106146566\\
3.49	0.0543805425658172\\
3.59	0.054482635913847\\
3.69	0.0562530476608563\\
3.79	0.0587029710990968\\
3.89	0.0579291219178157\\
3.99	0.0609822632257003\\
4.09	0.0583434926702919\\
4.19	0.0602938984184703\\
4.29	0.063045688626201\\
4.39	0.0712133419975329\\
4.49	0.0911631766888282\\
4.59	0.0959139540592672\\
4.69	0.0944104886517608\\
4.79	0.0941239048589746\\
4.89	0.0884936339081393\\
4.99	0.100368132717935\\
5.09	0.106241932342863\\
5.19	0.106984870204723\\
5.29	0.122463298865163\\
5.39	0.134339527836899\\
5.49	0.120117214454828\\
5.69	0.156026903190526\\
5.79	0.164196758962829\\
5.89	0.167783957869003\\
5.99	0.181735715261848\\
6.09	0.191557737771208\\
6.19	0.184786082111599\\
6.29	0.183997120531837\\
6.39	0.168906650924958\\
6.49	0.191840640557053\\
6.59	0.213126343488706\\
6.69	0.225586437305885\\
6.79	0.187209412129066\\
6.89	0.208441632125418\\
6.99	0.211631455050722\\
7.09	0.258143617799947\\
7.19	0.288834376566408\\
7.29	0.268354687004269\\
7.39	0.313713302586826\\
7.49	0.333823355459419\\
7.59	0.404524764254806\\
7.69	0.376651641270465\\
7.79	0.386269256784986\\
7.89	0.416838277991516\\
7.99	0.535028560153258\\
8.09	0.759220773762276\\
8.19	0.619754891805585\\
8.29	0.750055196756174\\
8.39	0.938704260992781\\
8.49	1.05817093458937\\
8.59	1.21826198128015\\
8.69	1.35058255041003\\
8.99	2.20841064407426\\
9.29	4.53055811169159\\
9.39	5.1232497600229\\
9.49	5.93943394165003\\
9.59	7.04488510200781\\
9.69	7.58770926219052\\
9.79	8.75353785198775\\
9.89	8.58270243844383\\
9.99	11.7234622667826\\
};

\addplot[area legend, draw=none, fill=colorkEDMD, fill opacity=0.4, forget plot]
table[row sep=crcr] {%
x	y\\
0	0\\
0.01	1.66792779771426e-05\\
0.02	2.60795994987267e-05\\
0.03	3.17291148696611e-05\\
0.04	7.28301798233345e-05\\
0.05	9.90616345625551e-05\\
0.06	0.000125540313443147\\
0.07	0.000132920532476431\\
0.08	0.000129962851913353\\
0.09	0.000144952492269027\\
0.19	0.000362866975359632\\
0.29	0.000721561968679226\\
0.39	0.00102647378130164\\
0.49	0.00139564931523126\\
0.59	0.00108645681859142\\
0.69	0.00178034257726138\\
0.79	0.00236811824946718\\
0.89	0.00264123989200805\\
0.99	0.00326526969706599\\
1.09	0.00364646341423316\\
1.19	0.00429922742121595\\
1.29	0.00453345131987966\\
1.39	0.00476640411483356\\
1.49	0.00474683153259412\\
1.59	0.00490845730818742\\
1.69	0.00526347195604637\\
1.79	0.00550814917529783\\
1.89	0.00587724232478568\\
1.99	0.00650481548003434\\
2.09	0.00650489351205326\\
2.19	0.00669944775977659\\
2.29	0.00728448704821388\\
2.39	0.00779021957677904\\
2.49	0.00793551838019646\\
2.59	0.00859170639203457\\
2.69	0.00905295370884939\\
2.79	0.00940235914812035\\
2.89	0.00991897119775138\\
2.99	0.0106892337052987\\
3.09	0.0104219218693605\\
3.19	0.0109169572753948\\
3.29	0.0122786404166128\\
3.39	0.0119319125529661\\
3.49	0.0123493759649592\\
3.59	0.0121937593584888\\
3.69	0.0124099413234894\\
3.79	0.0126528203097118\\
3.89	0.0128995224311477\\
3.99	0.014058990034488\\
4.09	0.0136951573219665\\
4.19	0.0142286205138863\\
4.29	0.0144528218082416\\
4.39	0.0152491512126313\\
4.49	0.0154980112071243\\
4.59	0.0161759715503907\\
4.69	0.0168989994317461\\
4.79	0.0172807474225206\\
4.89	0.0182800755570893\\
4.99	0.0182876444453313\\
5.09	0.01871660223607\\
5.19	0.0184267499519772\\
5.29	0.0189261176399855\\
5.39	0.0192059146931028\\
5.49	0.0201565290958547\\
5.59	0.0181973289150194\\
5.69	0.0212014573260416\\
5.79	0.0199587765084687\\
5.89	0.0165723126242623\\
5.99	0.0147942390357149\\
6.09	0.0136251413506692\\
6.19	0.0162768983607037\\
6.29	0.0191194324442505\\
6.39	0.0195685520335481\\
6.49	0.0231401616860751\\
6.59	0.0191300865133844\\
6.69	0.0231077956268918\\
6.79	0.022848042025791\\
6.89	0.0202442981337129\\
6.99	0.0244792269284628\\
7.09	0.0254720922332952\\
7.19	0.028975714746103\\
7.29	0.0292013390823887\\
7.39	0.0296647218171415\\
7.49	0.0302039687237636\\
7.59	0.0306105859944366\\
7.69	0.03096819002974\\
7.79	0.0310756351715166\\
7.89	0.0316739350416747\\
7.99	0.0321849728357281\\
8.09	0.0328217879841915\\
8.19	0.0333143081179099\\
8.29	0.0340556056172886\\
8.39	0.0346918819218836\\
8.49	0.0353045983143413\\
8.59	0.0358711179240526\\
8.69	0.0359465880366679\\
8.79	0.036378152526419\\
8.89	0.0366350763672707\\
8.99	0.0374647781651833\\
9.09	0.0379845948455873\\
9.19	0.0383077009393208\\
9.29	0.0387258884756577\\
9.39	0.0393711227748183\\
9.49	0.0398215285516052\\
9.59	0.0405911683337646\\
9.69	0.0408651238659927\\
9.79	0.0411351070613454\\
9.89	0.0411981485456461\\
9.99	0.0409801661110236\\
9.99	0.213499410264328\\
9.89	0.216477569131646\\
9.79	0.217977830361748\\
9.69	0.21195839288498\\
9.59	0.207944771363058\\
9.49	0.206178235208327\\
9.39	0.203061673254516\\
9.29	0.207373884302682\\
9.19	0.20441977471971\\
9.09	0.193840627589702\\
8.99	0.192415484777996\\
8.89	0.189040858154869\\
8.79	0.183110506384848\\
8.69	0.187284139713527\\
8.59	0.186998564169075\\
8.49	0.184029716127247\\
8.39	0.189700618270436\\
8.29	0.177522401774716\\
8.19	0.172939616359083\\
8.09	0.16442570408824\\
7.99	0.153280945463791\\
7.89	0.145055716822221\\
7.79	0.143495678185937\\
7.69	0.141731106157853\\
7.59	0.136925758501416\\
7.49	0.135702665445946\\
7.39	0.123024499359292\\
7.29	0.115010925863684\\
7.19	0.110611286418073\\
7.09	0.108555664568652\\
6.99	0.104665911933085\\
6.89	0.0983328013908805\\
6.79	0.0926187975421739\\
6.69	0.0902258422015154\\
6.59	0.0848544616124463\\
6.49	0.0816390397097838\\
6.39	0.0775988497481566\\
6.29	0.0751669031581983\\
6.19	0.0728191886209476\\
6.09	0.0688919941480214\\
5.99	0.0680494226287523\\
5.89	0.0641818203770672\\
5.79	0.0632973855382822\\
5.69	0.0605274989415302\\
5.59	0.0594084062263283\\
5.49	0.0585910602778805\\
5.39	0.0581983354213542\\
5.29	0.0559217957082719\\
5.19	0.0540470125110331\\
5.09	0.0528444232228844\\
4.99	0.052957746505791\\
4.89	0.0509710182949336\\
4.79	0.0503295748078713\\
4.69	0.0494749972883985\\
4.59	0.0487053468277428\\
4.49	0.0469994767750978\\
4.39	0.0474022376585544\\
4.29	0.044843701255304\\
4.19	0.0416152362965064\\
4.09	0.0397842510041932\\
3.99	0.0369992299707896\\
3.89	0.0339715210306278\\
3.79	0.0337678738782274\\
3.69	0.034983915616769\\
3.59	0.032714249954617\\
3.49	0.0315101256724927\\
3.39	0.031679843816112\\
3.29	0.0291251922313693\\
3.19	0.0296363474367892\\
3.09	0.0287354238968824\\
2.99	0.0257027950727338\\
2.89	0.0235391574999884\\
2.79	0.0236346800469737\\
2.69	0.0228533250859786\\
2.59	0.0220413934132634\\
2.49	0.0213407387168397\\
2.39	0.0195366015007839\\
2.29	0.0187625912266726\\
2.19	0.0175295498522192\\
2.09	0.0176152397766287\\
1.99	0.0168070039261543\\
1.89	0.0159368955738235\\
1.79	0.0152501555089743\\
1.69	0.0141036141807169\\
1.59	0.0126034449102743\\
1.49	0.0118279090062655\\
1.39	0.0117979165973813\\
1.29	0.0110490717463319\\
1.19	0.0115573548115192\\
1.09	0.0101621752775354\\
0.99	0.0104050523741937\\
0.89	0.0100640662008173\\
0.79	0.00863833766387242\\
0.69	0.00801252101299638\\
0.59	0.00751383581606642\\
0.49	0.00570932784664013\\
0.39	0.00503688893508175\\
0.29	0.0034056283841867\\
0.19	0.00200894246916359\\
0.09	0.00110286891101202\\
0.08	0.000970812680300728\\
0.07	0.00104951819445793\\
0.06	0.000830198541057294\\
0.05	0.00114416009827306\\
0.04	0.000690519230615043\\
0.03	0.000295964325954505\\
0.02	0.000320850765527273\\
0.01	0.000323148617569429\\
0	0\\
}--cycle;
\addplot [color=colorkEDMD, thick]
  table[row sep=crcr]{%
0.00999999999999979	0.000112405221161134\\
0.0500000000000007	0.000400023584446947\\
0.0600000000000005	0.000474926297887508\\
0.0700000000000003	0.000522260098954094\\
0.0800000000000001	0.000595975723280491\\
0.0899999999999999	0.000591037406351738\\
0.19	0.00107838193086548\\
0.289999999999999	0.00176636277600141\\
0.390000000000001	0.00241811506091817\\
0.49	0.00287502398804558\\
0.59	0.00357685516686012\\
0.69	0.00413077036396233\\
0.790000000000001	0.00456998364759112\\
0.890000000000001	0.00519940910818319\\
0.99	0.00571253217238356\\
1.09	0.00613494495138121\\
1.19	0.00709543434462381\\
1.29	0.00777780698258876\\
1.39	0.00821937546753288\\
1.49	0.00929152585892098\\
1.59	0.00991263386794225\\
1.69	0.0106796838884228\\
1.79	0.010785350775447\\
1.89	0.0117682407117782\\
1.99	0.012020144997512\\
2.09	0.0129467852302901\\
2.19	0.0134562646085162\\
2.29	0.0138161524808591\\
2.39	0.0144993166336597\\
2.69	0.0160162012388838\\
2.79	0.0169377382366095\\
2.89	0.0169391054869605\\
3.09	0.0187472666380447\\
3.19	0.0194504233844341\\
3.29	0.0197035321925822\\
3.39	0.0205356393749189\\
3.49	0.0215488304270423\\
3.59	0.0221091761907132\\
3.79	0.0235000223992979\\
3.89	0.0232677491827217\\
3.99	0.0248627161284651\\
4.09	0.0248288249185413\\
4.29	0.0263663945976717\\
4.39	0.0280108784575117\\
4.69	0.02937861067505\\
4.79	0.0294708800769084\\
5.19	0.0319631881435467\\
5.29	0.0328461707137444\\
5.49	0.0338973026725059\\
5.59	0.0342419522385409\\
5.69	0.0350070699933774\\
5.79	0.0354106216760618\\
5.89	0.0364219773783918\\
5.99	0.0370732658800148\\
6.09	0.0374518777967248\\
6.19	0.0386077753486569\\
6.29	0.039462340499105\\
6.49	0.0427142582945832\\
6.59	0.0433216100546126\\
6.69	0.0456612836776693\\
6.79	0.0476529609768662\\
6.89	0.0493030398146052\\
6.99	0.0519780304470847\\
7.09	0.0540413031763351\\
7.29	0.0555189268869893\\
7.39	0.0575466393647172\\
7.49	0.0609163021440571\\
7.89	0.0653094463267586\\
7.99	0.0662226991071314\\
8.09	0.0688371438313964\\
8.19	0.0702758298751054\\
8.29	0.0711647810286329\\
8.39	0.0741891814919941\\
8.49	0.0740480006231816\\
8.59	0.0775287641358312\\
8.69	0.0793471078466297\\
8.79	0.078236780073873\\
8.89	0.0806370190512483\\
9.29	0.0850851439585446\\
9.39	0.0851671455687883\\
9.49	0.086816582628048\\
9.89	0.0894067340450042\\
9.99	0.0918159509266702\\
};

\end{axis}
\end{tikzpicture}%

%% file: closed-loop-performance-max.tex
%
%
\begin{tikzpicture}
\begin{axis}[%
width=0.9\linewidth,
height=0.4\linewidth,
xmin=0,
xmax=80,
xlabel style={font=\color{white!15!black},yshift=3pt},
xlabel={$t$},
ymode=log,
ymin=0.5,
ymax=2000,
ytick distance=100,
yminorticks=true,
ylabel style={font=\color{white!15!black}},
ylabel={\small\hspace*{-0.05\linewidth}$\sum\limits_{k=0}^{\lfloor t/\Delta t\rfloor} \|x_k\| + \|u_k\|$},
xmajorgrids,
ymajorgrids,
yminorgrids,
]
\addplot [color=colorSafEDMD, ultra thick]
  table[row sep=crcr]{%
0	0.559243410586283\\
0.0100000000000051	1.11819244107265\\
0.019999999999996	1.6768472366787\\
0.0300000000000011	2.23520794249955\\
0.0499999999999972	3.35104766454763\\
0.0699999999999932	4.46571276551263\\
0.0900000000000034	5.57920440174057\\
0.189999999999998	11.1291009499278\\
0.290000000000006	16.6498332451582\\
0.390000000000001	22.1415437555469\\
0.489999999999995	27.6043738754417\\
0.590000000000003	33.0384639811807\\
0.689999999999998	38.4439534840537\\
0.790000000000006	43.8209808804808\\
0.890000000000001	49.1696837996953\\
0.989999999999995	54.4901990495364\\
1.09	59.7826626603823\\
1.19	65.0472099274132\\
1.29000000000001	70.2839754514822\\
1.39	75.4930931788496\\
1.59	85.8289179883221\\
1.79000000000001	96.0557442988946\\
1.98999999999999	106.174624027318\\
2.19	116.186602269762\\
2.39	126.092718465406\\
2.59	135.894007620972\\
2.89	150.401647259919\\
3.19	164.679223393925\\
3.48999999999999	178.730213064104\\
3.79000000000001	192.558103429662\\
4.19	210.654307198178\\
4.59	228.368488849421\\
5.09	249.986847849645\\
5.59	271.038019138571\\
6.19	295.573905291052\\
6.79000000000001	319.344574331112\\
7.48999999999999	346.145274203399\\
8.19	371.979757930568\\
8.98999999999999	400.366780979371\\
9.89	430.90545626038\\
10.89	463.166862628397\\
11.99	496.695009159428\\
13.19	531.008515650899\\
14.49	565.605587231959\\
15.89	599.975118831019\\
17.39	633.61729045666\\
19.09	667.988390696216\\
20.99	702.030025971189\\
23.09	734.789818234213\\
25.39	765.458023207362\\
27.99	794.402189033573\\
30.89	820.578387032294\\
34.19	843.875158652356\\
37.99	863.820950178453\\
42.49	880.074197149701\\
47.99	892.156411910069\\
55.39	900.227449109265\\
67.59	904.680747905164\\
79.99	905.726222736402\\
};
\label{plot:SafEDMD}

\addplot [color=colorSafEDMDnI, ultra thick]
  table[row sep=crcr]{%
0	1.04316064735885\\
0.0100000000000051	2.0763374576507\\
0.019999999999996	3.09959280418332\\
0.0300000000000011	4.11298932731958\\
0.0499999999999972	6.11045827645727\\
0.0699999999999932	8.06925366695364\\
0.0900000000000034	9.9898936247443\\
0.189999999999998	19.0396890266711\\
0.290000000000006	27.2202322329907\\
0.390000000000001	34.6096202696232\\
0.489999999999995	41.2897738032519\\
0.590000000000003	47.3412254385734\\
0.689999999999998	52.8386454144715\\
0.790000000000006	57.8483546882404\\
0.890000000000001	62.4277133411164\\
0.989999999999995	66.6256787238992\\
1.09	70.483872444078\\
1.19	74.037746784556\\
1.29000000000001	77.3176566396011\\
1.48999999999999	83.1568180864896\\
1.69	88.1729043700925\\
1.89	92.4989092083883\\
2.09	96.2405967992153\\
2.39	100.941205708485\\
2.69	104.747046926296\\
2.98999999999999	107.835274058551\\
3.39	111.072520704696\\
3.89	114.047502123727\\
4.48999999999999	116.502108449442\\
5.19	118.343236143231\\
6.19	119.830000456479\\
7.59	120.764860584091\\
10.19	121.250279668395\\
19.79	121.349335821694\\
79.99	121.34947775079\\
};
\label{plot:SafEDMD-monomials}

\addplot [color=colorkEDMD, ultra thick]
  table[row sep=crcr]{
0	1.37752908755645\\
0.0100000000000051	2.73650055813074\\
0.019999999999996	4.07681935806293\\
0.0300000000000011	5.39838708242768\\
0.0499999999999972	7.98484444226839\\
0.0699999999999932	10.4949075753222\\
0.0900000000000034	12.9273663157046\\
0.189999999999998	23.8909751441132\\
0.290000000000006	33.0360243817253\\
0.390000000000001	40.8604554667833\\
0.489999999999995	47.6987783419615\\
0.590000000000003	53.7209153991614\\
0.689999999999998	59.0031011551174\\
0.790000000000006	63.5683950441659\\
0.890000000000001	67.4367783415745\\
0.989999999999995	70.66628924729\\
1.09	73.3459774445468\\
1.19	75.5640748580605\\
1.29000000000001	77.3946202325609\\
1.48999999999999	80.1301742599332\\
1.69	81.9471632656734\\
1.89	83.1387114713561\\
2.19	84.1939176320052\\
2.69	84.9576682769591\\
3.59	85.2930158766735\\
7.48999999999999	85.3487134845365\\
79.99	85.3487254869809\\
};
\label{plot:kEDMD}

\end{axis}
\end{tikzpicture}%

%% file: sec5-conclusion.tex
%
%
\section{CONCLUSION}\label{sec:conclusion}
We derived novel deterministic error bounds on the full approximation error of data-based bilinear surrogate models for unknown nonlinear systems. 
By leveraging kernel-based dictionaries that are invariant by construction, our approach circumvents restrictive assumptions on the projection error required by existing methods. 
The derived error bounds are state- and input-dependent, offering direct applicability to Koopman-based robust controller design with closed-loop guarantees.

%% file: sec6-app-proof-error-bounds.tex
%
%
\bgroup
\renewcommand\thesection{Appendix A}
\section{PROOF OF THEOREM~\ref{thm:error-bound}}\label{sec:app-proof-of-theorem}
\egroup

We divide the proof into two parts. 
First, we derive an approximate control-affine representation of the nonlinear propagation step $\Phi(x^+)$ while characterizing its approximation error.
Second, we employ the kernel-based interpolation~\eqref{eq:psi-dynamics-kEDMD} to deduce the proposed bilinear surrogate model and show that this surrogate admits a bounded residual error.

\noindent\textbf{Part I -- Approximate control-affine representation:~}
First, we characterize the nonlinear propagation step $\Phi(x^+)$ via 
\newlength\mylenA
\settoheight\mylenA{$
    \Phi(x^+) = \Phi(f(x)) 
    + \sum_{i=1}^m \left[
        \Phi(\tilde{g}_i(x)) - \Phi(f(x))
    \right] u_i 
    + r_\Phi(x,u)
    \hspace*{0.12\linewidth}
$}
\begin{equation}\label{eq:surrogate-control-affine}
\resizebox{\linewidth}{\mylenA}{$
    \hspace*{-0.03\linewidth}
    \Phi(x^+) = \Phi(f(x)) 
    + \sum_{i=1}^m \left[
        \Phi(\tilde{g}_i(x)) - \Phi(f(x))
    \right] u_i 
    + r_\Phi(x,u)
    \hspace*{0.12\linewidth}
$}
\hspace*{-0.12\linewidth}
\end{equation}
which is approximately control-affine.
Note that $r_\Phi(x,u)$ depends on $\Delta t$ through the definition of $f$, $\tilde{g}_i$ in~\eqref{eq:dynamics-bilinear-surrogate-approx}.
In the following, we derive an error bound on $r_\Phi$ by investigating the propagation step~\eqref{eq:surrogate-control-affine} element-wise.
To this end, let $j\in[1:d]$ and denote $(r_\Phi(x,u))_j$ by $h_j(\Delta t)$ for simplicity.
Then, the Taylor expansion of $h_j(\Delta t)$ around $\Delta t=0$ yields
\begin{equation*}
    \textstyle
    h_j(\Delta t) 
    = h_j(0) 
    + \Delta t \left.\frac{\partial h_j(\Delta t)}{\partial \Delta t}\right|_{\Delta t=0} 
    + \frac{\Delta t^2}{2} \left.\frac{\partial^2 h_j(\Delta t)}{\partial \Delta t^2}\right|_{\Delta t=\tau} 
\end{equation*}
for some $\tau\in[0,\Delta t]$. 
Further, we define $\tilde{G}_c(x,u) = f_c(x) + G_c(x)u$ and $\tilde{g}_{ci}(x)=f_c(x) + g_{ci}(x)$, $i\in[1:m]$.
Recall~\eqref{eq:functions-discretized} to observe
$
    h_j(0) 
    = \phi_{x_j}(x) - \phi_{x_j}(x) - \sum_{i=1}^m \left[ \phi_{x_j}(x) - \phi_{x_j}(x) \right]u_i
    = 0
$
and
\begin{align}
    &\textstyle
        \frac{\partial h_j(\Delta t)}{\partial \Delta t}\bigg|_{\Delta t=0}
        = \frac{\partial \phi_{x_j}\big(x + \int_{0}^{\Delta t} \tilde{G}_c(x(t),u)\dd t\big)}{\partial x}\bigg|_{\Delta t=0}
        \tilde{G}_c(x,u)
    \nonumber\\
    &\textstyle
        \quad - \frac{\partial \phi_{x_j}(f(x))}{\partial x}\bigg|_{\Delta t=0} f_c(x)
        -\sum_{i=1}^m \frac{\partial \phi_{x_j}(\tilde{g}_i(x))}{\partial x}\bigg|_{\Delta t=0} \tilde{g}_{ci}(x) u_i
    \nonumber\\
    &\textstyle
    \quad +\sum_{i=1}^m \frac{\partial \phi_{x_j}(f(x))}{\partial x}\bigg|_{\Delta t=0} f_c(x) u_i
    = 0.
\end{align}
Further, we derive
\newlength\mylenG
\settoheight\mylenG{$
    \quad +\ \tilde{G}_c(x,u)^\top
    \nabla^2 \phi_{x_j}\Big(x + \int_{0}^{\Delta t} \tilde{G}_c(x(t),u)\dd t\Big)
    \Big|_{\Delta t=\tau} 
    \tilde{G}_c(x,u)
$}
\begin{align}
    &\textstyle
        \frac{\partial^2 h_j(\Delta t)}{\partial \Delta t^2}\bigg|_{\Delta t=\tau} = \sum_{i=1}^m f_c(x)^\top 
        \nabla^2 \phi_{x_j}\left(f(x)\right)\big|_{\Delta t=\tau} 
        f_c(x) u_i
    \nonumber\\
    &\resizebox{\linewidth}{\mylenG}{$
        \quad +\ \tilde{G}_c(x,u)^\top
        \nabla^2 \phi_{x_j}\Big(x + \int_{0}^{\Delta t} \tilde{G}_c(x(t),u)\dd t\Big)
        \Big|_{\Delta t=\tau} 
        \tilde{G}_c(x,u)
    $}
    \nonumber\\
    &\textstyle
    \;\;\, - f_c(x)^\top 
    \nabla^2 \phi_{x_j}\left(f(x)\right)\big|_{\Delta t=\tau} 
    f_c(x)
    \nonumber\\
    &\textstyle
    \;\;\, -\sum_{i=1}^m \tilde{g}_{ci}(x)^\top 
    \nabla^2 \phi_{x_j}\left(\tilde{g}_i(x)\right)\big|_{\Delta t=\tau} 
    \tilde{g}_{ci}(x) u_i.
\end{align}
Then, exploiting~\eqref{eq:assumption-second-derivative} yields 
$
    \left\|\frac{2h_j(\Delta t)}{\Delta t^2}\right\|
$
is less than or equal to
\begin{align}
    &\textstyle
    \max_{x\in\bbX}\left\{\left\|\nabla^2 \phi_{x_j}(x)\right\|\right\} 
    \Big( 
        \|\tilde{G}_c(x,u)\|^2 
    \nonumber\\&\qquad\textstyle
        + \|f_c(x)\|^2 \big\| 1 - \sum_{i=1}^m u_i \big\|
        + \sum_{i=1}^m \left\|\tilde{g}_{ci}(x)\right\|^2\left|u_i\right|
    \Big)
    \nonumber\\
    &\textstyle
    \leq D_\phi  
    \Big( 
        \|f_c(x)\|^2+2\|f_c(x)^\top G_c(x)u\|+\|G_c(x)u\|^2 
    \nonumber\\&\qquad\textstyle
        + \|f_c(x)\|^2 \big\|1 - \sum_{i=1}^m u_i \big\|
        + \sum_{i=1}^m \big(
            \|f_c(x)\|^2|u_i|
    \nonumber\\&\qquad\textstyle
            +2\|f_c(x)^\top g_{ci}(x)\||u_i| 
            +\|g_{ci}(x)\|^2|u_i|
        \big)
    \Big)
    \nonumber\\
    &\textstyle
    \leq D_\phi 
    \bigg( 
        \|f_c(x)\|^2\Big( 
            1 + \Big\|1 - \sum_{i=1}^m u_i\Big\| + \sum_{i=1}^m |u_i|
        \Big)
    \nonumber\\&\qquad\textstyle
        + 2\|f_c(x)\|\Big(\|G_c(x)\|\|u\| + \sum_{i=1}^m \|g_{ci}(x)\||u_i|\Big)
    \nonumber\\&\qquad\textstyle
       +\|G_c(x)\|^2\|u\|^2 + \sum_{i=1}^m \|g_{ci}(x)\|^2|u_i|
    \bigg).
\end{align}
Further, observe $
    \sum_{i=1}^m \|g_{ci}(x)\| |u_i|
    \leq \|G_c(x)\| \|u\|,
$
\begin{equation*}
    \textstyle
     \sum_{i=1}^m \|g_{ci}(x)\|^2 |u_i|
     \leq \sum_{i=1}^m \|g_{ci}(x)\|^2 \|u\| 
     = \|G_c(x)\|^2 \|u\|,
\end{equation*}
and define $\tilde{u} = \max_{u\in\bbU} \big\{\big\|1 - \sum_{i=1}^m u_i\big\|\big\}$.
Then, due to the Lipschitz continuity of $f_c$ with Lipschitz constant $L_f$, $f_c(0)=0$, and using $\|G_c(x)\|\leq \bar{G}$ for all $x\in\bbX$, we obtain
\begin{multline}
    \textstyle
    \left\|\frac{2h_j(\Delta t)}{\Delta t^2}\right\| 
    \leq 
    D_\phi \Big( 
         L_f^2 \big(1 + \tilde{u} + \max_{u\in\bbU}\|u\|_1\big)
        \|x\|^2
    \\\textstyle
        + 4 L_f \bar{G} \|x\| \|u\|
        + \bar{G}^2 (\|u\|^2 + \|u\|)
    \Big).
\end{multline}
Finally, we combine this bound for each $j\in[1:d]$ to obtain
\begin{align*}
    \textstyle
    \|r_\Phi(x,u)\| 
    &\textstyle
    = \left\|
        \begin{bmatrix}
            h_1(\Delta t) & \cdots & h_d(\Delta t)  
        \end{bmatrix} 
    \right\|
    \nonumber\\
    &\textstyle
    \leq \sqrt{d}\Delta t^2 \big[
        \tfrac{D_\phi}{2} L_f^2 \big(1 + \tilde{u} + \max_{u\in\bbU} \|u\|_1 \big) \|x\|^2
    \nonumber\\&\textstyle
        + 2 D_\phi L_f \bar{G} \|x\| \|u\|
        + \tfrac{D_\phi}{2} \bar{G}^2 (\|u\| + \|u\|^2)
    \big].
\end{align*}

\noindent\textbf{Part II -- Kernel-based interpolation:~}
Next, we use the collected data to represent the individual terms in~\eqref{eq:surrogate-control-affine}.

\emph{a) Term $\Phi(f(x))$:~}
We define the action of the Koopman operator $\cK_0$ corresponding to the (autonomous) dynamics $x^+=f(x)$ as $(\cK_0 \Phi)(x) \coloneqq \Phi(f(x))$ with $\Phi$ defined in~\eqref{eq:features-Phi}.
Since $\cK_0$ is an infinite-dimensional operator, we approximate its action via kEDMD. 
More precisely, the Koopman action $\cK_0\phi_{x_j}$ of any element $\phi_{x_j}\in\cN_s$ in $\Phi$ is approximated by
\begin{equation}
    \textstyle
    \phi_{x_j}(f(x))
    = (\cK_0 \phi_{x_j})(x)
    \approx
    \sum_{j=1}^{d} (\hat{K}_0\phi_{x_j,\cX})_j \phi_{x_j}(x)
\end{equation}
using the kEDMD approximation~\eqref{eq:psi-dynamics-kEDMD}, where we substitute $F$ and $\psi$ by $f$ and $\phi_{x_j}$, respectively.
By stacking all elements $\phi_{x_j}$, we obtain the vectorized representation $\Phi(f(x)) \approx (\hat{K}_0 \Phi_\cX)^\top \Phi(x)$ with $
    \Phi_\cX = \begin{bmatrix}
        \phi_{x_1,\cX} & \cdots & \phi_{x_d,\cX}
    \end{bmatrix}
    = K_\cX
$.
Further, exploiting $\Phi_\cX^\top = K_\cX$ due to the symmetry of $K_\cX$ and the definition of $\hat{K}_0$ according to~\eqref{eq:psi-dynamics-kEDMD} yields
\begin{equation*}
    (\hat{K}_0\Phi_{\cX})^\top 
    = K_\cX K_\cX^{-1} K_{f(\cX)}^\top K_\cX^{-1} 
    = K_{f(\cX)}^\top K_\cX^{-1}
    = A
\end{equation*}
in line with the definition of $A$ in~\eqref{eq:dynamics-bilinear-surrogate-approx}.
Here, we consider $A$ as a perturbed matrix representation of $P_\cX \cK_0$ in the canonical basis of the dictionary given by $\{\phi_{x_j}\,|\,j\in [1:d]\}$, where $P_\cX$ is the orthogonal projection onto $\operatorname{span}(\phi_{x_1},\ldots,\phi_{x_d})$.
The perturbation is in $\cO(\Delta t^2)$ and occurs since the collected data samples $\{x_{jl},u_{jl},x_{jl}^+\}_{l=1}^{d_j}$ for each $j\in[1:d]$ do not follow the approximation~\eqref{eq:functions-discretized} but the true sampled dynamics~\eqref{eq:dynamics-nonlinear-sampled}.
Next, we derive a rigorous error bound to
\begin{multline}\label{eq:proof-perturbed-residual-A}
    \Phi(f(x)) - A\Phi(x) =
    (\cK_0\Phi - A\Phi)(x) 
    \\
    = (\cK_0\Phi - P_\cX \cK_0 \Phi)(x) + (P_\cX \cK_0 \Phi - A\Phi)(x)
\end{multline}
by providing a bound on the individual entries.
To this end, let $j\in [1:d]$. 
Then,~\cite[Theorem 3.7]{bold:philipp:schaller:worthmann:2024} yields\footnote{We use~\cite[Theorem 3.7]{bold:philipp:schaller:worthmann:2024} with relaxed assumptions on $s$ by exploiting~\cite[Theorem 11.17]{wendland:2004} instead of~\cite[Lemma 2.5]{bold:philipp:schaller:worthmann:2024} in its proof.}
\begin{equation*}
    |(\cK_0 \phi_{x_j} - P_\cX\cK_0 \phi_{x_j})(x)| 
     \leq C_1h_\cX^{s-\nicefrac{1}{2}} \operatorname{dist}(x,\cX) \|\phi_{x_j}\|_{\cN_s}
\end{equation*}
for fill distance $h_\cX \leq h_0$ and smoothness degree $s$ of the Wendland kernels, where we refer to~\cite[Section 3.3]{wendland:2004} for the derivation of the constants $C_1$, $h_0$ only depending on the domain $\bbX$.
Further, $0\in\cX$ yields $\operatorname{dist}(x,\cX) \leq \|x\|$.
For the second term in~\eqref{eq:proof-perturbed-residual-A}, we observe $x_{jl}^+ = (f(x_j) + \xi_f(x_j,u_{jl})) + (G(x_j) + \xi_G(x_j,u_{jl})) u_{jl}$, where $
    \|\xi_f(x_j,u_{jl})\|
    \leq \Delta t^2L_f \hat{C}_3
$ and $
    \|\xi_G(x_j,u_{jl})\| 
    \leq \Delta t^2L_G \hat{C}_3
$
for $\hat{C}_3 = \frac{1}{2} (L_f \bar{x} + \bar{G} \bar{u})$.
Hence, we obtain the approximately control-affine system
\begin{equation}\label{eq:dynamics-nonlinear-sampled-control-affine}
    x_{jl}^+ 
    = f(x_j) + G(x_j)u_{jl} + \xi(x_j,u_{jl}),
\end{equation}
where $\xi(x_j,u_{jl}) = \xi_f(x_j,u_{jl}) + \xi_G(x_j,u_{jl}) u_{jl}$ is bounded by $ \|\xi(x_j,u_{jl})\| \leq \Delta t^2\tilde{C}_3$ with 
$
    \tilde{C}_3
    = \hat{C}_3 (L_f + L_G \bar{u})
$.
Now, we exploit~\cite{ziemann:tsiamis:lee:jedra:matni:pappas:2023} to bound the deviation of $\hat{H}_j$ from the true values
$
    H_j = \begin{bmatrix}
        f(x_j) & G(x_j)
    \end{bmatrix}
$ 
as
\begin{equation*}
    \textstyle
    \| \hat{H}_j - H_j \| 
    \leq \frac{\sigma_{\max}\left(
        \begin{bmatrix}
            \xi(x_j,u_{j1}) & \!\!\cdots\!\!\! & \xi(x_j,u_{jd_j})
        \end{bmatrix}\right)
    }{\sigma_{\min}(\bar{U}_j)
    }
    \leq \Delta t^2 C_3
\end{equation*}
with 
$
    C_3 = \tilde{C}_3 \max_{j\in[1:d]} \big\{ \nicefrac{\sqrt{d_j}}{\sigma_{\min}(\bar{U}_j)}\big\}
$ 
and $\bar{U}_j$ as in~\eqref{eq:rank-condition-inputs}, where $\sigma_{\min}(\bar{U}_j)$ is positive due to Assumption~\ref{ass:input-rank-condition}.
Then, exploiting $f(0)=0$ with similar techniques as in~\cite[Theorem 4.3]{bold:philipp:schaller:worthmann:2024} yield 
\begin{equation}
    \textstyle
    \hspace*{-0.02\linewidth}
    \|(P_{\mathcal{X}} \mathcal{K}_0 \phi_{x_j})(x) - (A \Phi(x))_j\| \leq \Delta t^2 C_2 C_3 \|K_\cX^{-1}\| \|x\|,
\end{equation}
see again~\cite[Section 3.3]{wendland:2004} for the constant $C_2>0$. 
Thus, by leveraging~\eqref{eq:proof-perturbed-residual-A}, we establish the kernel-based representation
\begin{equation}\label{eq:kernel-interpolation-A}
    \textstyle
    \Phi(f(x)) = A \Phi(x) + r_A(x),
\end{equation}
where with $
    \|\Phi\|_{\mathcal{N}_s} := \| (\|\phi_{x_1}\|_{\cN_s}, \ldots, \|\phi_{x_d}\|_{\cN_s})\|
$ holds
\begin{align}
    \|r_A(x)\|
    &= \|\Phi(f(x)) - A \Phi(x)\|
    = \|(\cK_0\Phi)(x) - A \Phi(x)\|
    \nonumber\\
    &\leq (C_1h_\cX^{s-\nicefrac{1}{2}}\|\Phi\|_{\cN_s} 
    + \sqrt{d} \Delta t^2 C_2 C_3 \|K_\cX^{-1}\| ) \|x\|.
\end{align}

\emph{b) Term $\Phi(\tilde{g}_i(x))$:~}
Analogous to the consideration of $\Phi(f(x))$, we define the action of the Koopman operator corresponding to the (autonomous) dynamics $x^+ = \tilde{g}_i(x)$ as $(\cK_i \Phi)(x) \coloneqq \Phi(\tilde{g}_i(x))$, $i\in[1:m]$. 
Hence, by replacing $f$ by $\tilde{g}_i$ in the steps above, we obtain
\begin{equation}\label{eq:kernel-interpolation-B_i}
    \textstyle
    \Phi(\tilde{g}_i(x)) = \hat{B}_i \Phi(x) + r_{\hat{B}_i}(x)
\end{equation}
with $\hat{B}_i = K_{\tilde{g}_i(\cX)}^\top K_\cX^{-1}$, where a bound on $r_{\hat{B}_i}$ follows analogously to the estimate of $r_A$.
More precisely,
$
    \|r_{\hat{B}_i}(x)\|
    \leq C_1h_\cX^{s-\nicefrac{1}{2}}\|\Phi\|_{\cN_s} \|x\| + \sqrt{d} \Delta t^2 C_2 C_3 \|K_\cX^{-1}\|
$, where the latter state-independent term remains to address the case $G(0)\neq 0$.

\emph{c) Propagation step $\Phi(x^+)$:~}
Finally, we combine the derived representations~\eqref{eq:kernel-interpolation-A} and~\eqref{eq:kernel-interpolation-B_i} of $\Phi(f(x))$ and $\Phi(\tilde{g}_i(x))$, $i\in[1:m]$, respectively, according to~\eqref{eq:surrogate-control-affine}, i.e.,
\begin{multline}\label{eq:propagation-step-Phi-individual-residuals}
    \textstyle
    \Phi(x^+)
    = A \Phi(x) + r_A(x)
    + \sum_{i=1}^{m} \left[\hat{B}_i \Phi(x) - A \Phi(x)\right] u_i
    \\\textstyle
    + \sum_{i=1}^{m} \left[r_{\hat{B}_i}(x) - r_A(x)\right] u_i + r_\Phi(x,u).
\end{multline}
In the following, we define the overall residual
\begin{equation}\label{eq:proof-definition-residual}
    \textstyle
    \hspace*{-0.015\linewidth}
    r(x,u) = r_A(x) + \sum_{i=1}^{m} \left[r_{\hat{B}_i}(x) - r_A(x)\right] u_i + r_\Phi(x,u) 
\end{equation}
and recall $B_i = (K_{\tilde{g}_i(\cX)}-K_{f(\cX)})^\top K_\cX^{-1} = \hat{B_i} - A$ in~\eqref{eq:dynamics-bilinear-surrogate-approx}, such that~\eqref{eq:propagation-step-Phi-individual-residuals} reads
\begin{multline*}
    \textstyle
    \Psi(x^+) + \Phi(0)
    = A \left(\Psi(x) + \Phi(0)\right) 
    \\\textstyle
    + \sum_{i=1}^{m} B_i \left(\Psi(x) + \Phi(0)\right) u_i
    + r(x,u),
\end{multline*}
where we substitute $\Phi(x) = \Psi(x) + \Phi(0)$.
Thus, exploiting $r_A(0) = 0$, i.e., $A\Phi(0) = (\cK_0 \Phi)(0) = \Phi(f(0)) = \Phi(0)$, establishes the bilinear surrogate dynamics~\eqref{eq:dynamics-bilinear-surrogate} for $B_0 = \begin{bmatrix}
    B_1\Phi(0) & \cdots & B_m\Phi(0)
\end{bmatrix}$.
Finally, leveraging the definition of $r$ in~\eqref{eq:proof-definition-residual} with the error bounds on $r_\Phi$, $r_A$, $r_{\hat{B}_i}$ yields
\begin{align*}
    \textstyle
    \|r(x,u)\|
    &\textstyle
    \leq \Big\|\Big(1-\sum_{i=1}^m u_i\Big)r_A(x)\Big\|
    \\&\quad\textstyle
    + \left\|\begin{bmatrix}
        r_{\hat{B}_1}(x) & \!\!\!\cdots\!\!\! & r_{\hat{B}_m}(x)
    \end{bmatrix}\right\| \|u\| 
    + \|r_\Phi(x,u)\|
    \\
    &\textstyle
    \leq \tilde{u} \|r_A(x)\|
    + \sqrt{m} \|r_{\hat{B}_i}(x)\| \|u\| 
    + \|r_\Phi(x,u)\|
\end{align*}
and, hence, the error bound~\eqref{eq:deterministic-bound-residual} with the constants~\eqref{eq:proportional-bound-residual-constants}. 
\qed